\documentclass[a4paper]{llncs}
\usepackage{lscape}
\usepackage{graphicx} 

\usepackage{amsmath,amssymb,amsbsy,amsfonts,latexsym,
            amsopn,amstext,amsxtra,euscript,amscd,color}
\usepackage{verbatim}
\usepackage{float}

\usepackage{caption}

\newtheorem{thm}{Theorem}
\newtheorem{lem}[thm]{Lemma}
\newtheorem{cor}[thm]{Corollary}

\newtheorem{prop}[thm]{Proposition}
\newtheorem{rem}[thm]{Remark}
\newtheorem{defn}[thm]{Definition}

\newtheorem{ex}[thm]{Example}

\def\00{{\bf 0}}
\def\11{{\bf 1}}

\newcommand\dsum{\displaystyle \sum}
\newcommand\dprod{\displaystyle \prod}
\newcommand\dbigoplus{\displaystyle \bigoplus}
\def\wt{{\rm wt}}

\usepackage[colorinlistoftodos,prependcaption,textsize=tiny]{todonotes}

\def \f  {\mathbb{F}}

\newcommand{\D}{\mathrm{D}}

\newcommand{\dd}{\mathrm{dd}}
\newcommand{\add}{\mathrm{add}}
\newcommand{\dt}{\mathrm{dt}}

\begin{document}
\title{Probabilistic estimation of the algebraic degree
of Boolean functions}


\author{ Ana S\u{a}l\u{a}gean \and Percy Reyes-Paredes}
\institute{Department of Computer Science, Loughborough University, UK\
\email{A.M.Salagean@lboro.ac.uk, A.P.Reyes-Paredes@lboro.ac.uk}}

\date{}

\maketitle

\begin{abstract}
The algebraic degree is an important parameter of Boolean functions used in cryptography. When a function in a large number of variables is not given explicitly in algebraic normal form, it might not be feasible to compute its degree. Instead, one can try to estimate the degree using probabilistic tests.

We propose a probabilistic test for deciding whether the algebraic degree of a Boolean function $f$ is below a certain value $k$.
The test involves picking an affine space of dimension
$k$ and testing whether the values on $f$ on that space sum up to zero. If $\deg(f)<k$, then $f$ will always pass the test, otherwise it will sometimes pass and sometimes fail the test, depending on which affine space was chosen. The probability of failing the proposed test is closely related to the number of monomials of degree $k$ in a polynomial $g$, averaged over all the polynomials $g$ which are affine equivalent to $f$.

We initiate the study of the probability of failing the proposed ``$\deg(f)<k$'' test. We show that in the particular case when the degree of $f$ is actually equal to $k$, the probability will be in the interval $(0.288788, 0.5]$, and therefore a small number of runs of the test is sufficient to give, with very high probability, the correct answer. Exact values of this probability  for all the polynomials in 8 variables were computed using the representatives listed by Hou and by Langevin and Leander.
\end{abstract}

\section{Introduction and motivation}
The algebraic degree is an important parameter of Boolean functions used in cryptography. A Boolean function $f$ in $n$ variables can be uniquely represented in ANF (algebraic normal form), i.e.\ as a polynomial in $n$ variables over $\f_2$ (the binary field) of degree at most one in each variable. The degree of this polynomial is called the algebraic degree of $f$. Ciphers which can be represented (or approximated) as functions of low degree are vulnerable to attacks such as higher order differential attacks.

When the ANF of $f$ is not given explicitly (e.g. $f$ is a composition of functions, or is given as a ``black box'') and depends on a large number of variables, it may not be  feasible to compute its degree. Instead, we aim to estimate the degree using probabilistic tests.

The coefficient of a particular monomial $x_{i_1}\cdots x_{i_k}$ of degree $k$ in the ANF  of $f$ can be computed by summing the values of $f$ over the vector space generated by the $k$ vectors $\mathbf{e_{i_1}}, \ldots, \mathbf{e_{i_k}}$ in the canonical basis. This method,  sometimes called the Moebius transform, has many applications in cryptography and coding theory; in cryptanalysis it was used to detect and exploit non-randomness features of the number of monomials of a given degree
(see \cite{Fil02,ONe07,DinSha09,Vie07}).

One could use the Moebius transform to estimate the degree of a function as follows: pick a monomial of degree $k$, compute its coefficient in $f$ and test whether it is zero.
If it is not, then we know that $\deg(f) \ge k$.
If we run this test for several monomials of degree $k$ and all the computed coefficients are zero, then we conclude that $f$ probably has degree strictly less than $k$.
 The probability of finding a monomial of degree $k$ (and correctly concluding $\deg(f) \ge k$) in one run of the test is equal to the proportion of monomials of degree $k$ that have non-zero coefficients in $f$. Therefore, this method has  the shortcoming that
 if $f$ has degree at least $k$ but only has a very small number of monomials of degree $k$, then one might  incorrectly classify $f$ as having degree less than $k$, as illustrated in the following example:
\begin{ex}\label{ex:old-test}
  The function $f(x_1,\ldots, x_9) =x_1x_2x_3 +   x_4x_5x_6 +x_7x_8x_9$ has only 3 out of the $\binom{9}{3} = 84$ possible monomials of degree 3 in 9 variables. Assume we run the test based on the Moebious transform, to search for monomials of degree 3. Each run of the test has a probability of only $\frac{3}{84}\approx 0.0357$ to detect a monomial. After running the test 9 times, for example, we have still a rather high  probability of $(1-\frac{3}{84})^9 \approx 0.72$ that no monomial of degree $3$ has been detected yet by the test, and therefore we might wrongly conclude that $\deg(f) < 3$.
\end{ex}

In this paper we generalise this idea. The intuition behind our proposed method to test whether $\deg(f) < k$ is that even if a function $f$ has a very small number of monomials of degree $k$, after applying a random affine invertible change of variables to $f$ (the degree of $f$ being invariant to such changes of variables), the number of monomials of degree $k$ is likely to be high and therefore it will be easier to probabilistically detect their existence. The aim is that the test should perform reasonably well for \emph{all} functions.

We will call our proposed test the $\deg(f)<k$ test and define it as follows. Pick
$u_0, u_1,\ldots, u_k \in \f_2^n$ and check whether the sum of the values of $f$ over the affine space $u_0 \oplus \langle u_1, \ldots, u_k \rangle$ is zero. Again, given a function $f$ we run the test several times. If it passes all test we conclude that $\deg(f)<k$, otherwise we conclude $\deg(f)\ge k$.
  A function $f$ of degree less than $k$ will always pass this test (there are no false negatives). A function of degree $k$ or more will sometimes fail and sometimes pass the test, depending on the chosen vectors.
We denote by $\dt_k(f)$ the probability of failing this test, taken over all values $u_0, u_1,\ldots, u_k \in \f_2^n$. This probability determines the probability $(1-\dt_k(f))^t$ of wrongly concluding, after $t$ tests, that $\deg(f)<k$ when in fact $\deg(f)\ge k$ (false positive). Ideally, $\dt_k(f)$ should not be very low. A very small value of $\dt_k(f)$ would mean that we would need to run the test a very large number of times to obtain a reasonable accuracy.

We initiate the study of the probability $\dt_k(f)$ of failing the $\deg(f)<k$ test. We consider the case when the degree of $f$ is in fact equal to $k$ (although we would not know this beforehand, if we knew we would not need to do any test). We prove in Theorem~\ref{thm:bounds} and Corollary~\ref{cor:bounds} that  the probability $\dt_k(f)$ satisfies an upper bound of 0.5 and a lower bound of 0.288788... (q-Pochhammer symbol at $(0.5, 0.5, \infty)$). This means there are no functions with very low probability $\dt_k(f)$, and therefore a small number of runs of the test is sufficient to give, with very high probability, the correct answer. For example, to a obtain a probability of less than 0.05 that a polynomial of degree $k$ has been incorrectly classified as being of degree less than $k$, we would need to run the test 9 times. Contrast this with the situation illustrated in Example~\ref{ex:old-test}.

We compute and analyse the values of the probability of failing the $\deg(f)<k$ test for all functions in 8 variables of degree $k$,
using the representatives listed by Hou~\cite{Hou96} and by Langevin and Leander~\cite{LanLea07} (see Section~\ref{sec:numeric}).

The study of the probability of failing the $\deg(f)<k$ test for polynomials of degree strictly higher than $k$ will be the subject of future work.

The probability $\dt_k(f)$ is connected to other existing notions as follows. For $k=2$, i.e. the $\deg(f)<2$  test, if we restrict to linear rather than affine spaces, we obtain the usual textbook linearity test $f(u_1 \oplus u_2) = f(u_1) \oplus f(u_2)$, often called the BLR test. The probability of failing the BLR test was studied in several papers, see for example~\cite{BCHKS96}.  In~\cite{WinSal15}, in the context of the cube/AIDA attack, we proposed a linearity test similar to the $\deg(f)<k$ test above, but fixing a linear space of dimension $k$ and running the $\deg(f)<m$ test on all its subspaces of dimension $m$, for all $2\le m \le k$.

We show in Theorem~\ref{thm:relation-between-the-two-notions} that the probability of failing the $\deg(f)<k$ test, when restricted to
affine spaces of dimension exactly $k$,
is equal to the average number of monomials of degree $k$ over all the polynomials in the affine equivalence class of $f$. We propose this average density of monomials of degree $k$ as a new parameter of Boolean functions (see Definition~\ref{def:add}). It is somewhat similar, but distinct from the notion of algebraic thickness defined in~\cite{Car02}, see Remark~\ref{rem:thickness}.

\section{Definitions}
We denote by $\f_2$ the finite field with 2 elements, and by $\f_2^n$ the $n$-dimensional vector space over $\f_2$. Addition in $\f_2$ and in $\f_2^n$ will be denoted by $\oplus$.

Any function $f:\f_2^n \rightarrow \f_2$ can be represented in its algebraic normal form (ANF), i.e. as a polynomial function given by a polynomial of degree at most 1 in each variable:
\[f(x_1, \ldots, x_n)  = \dbigoplus  _{a_1, \ldots, a_n \in \f_2} b_{(a_1, \ldots, a_n)} x_1^{a_1}\cdots x_n^{a_n},\]
with $b_{(a_1, \ldots, a_n)} \in \f_2$.
The degree of this polynomial is called the algebraic degree of $f$, and here we will call it simply the degree of $f$ and denote it by $\deg(f)$.

The coefficients of the ANF of $f$ can be computed by the following formula (see, for example, \cite[Chapter~13, Theorem~1]{McWS}), which is sometimes called the Moebius transform:
\begin{equation}\label{eq:Moebius1}
  b_{(a_1, \ldots, a_n)} = \dbigoplus  _{x_1 \le a_1, \ldots, x_n \le a_n}f(x_1, \ldots, x_n).
\end{equation}
An equivalent form of this formula can be obtained as follows: let $\{i_1,i_2,\ldots ,i_k \}$ be the support of $a$, i.e. $a_i=1$ if and only if $i \in \{i_1,i_2,\ldots ,i_k \}$.
In other words $b_{(a_1, \ldots, a_n)}$ is the coefficient of $x_{i_1}x_{i_2}\cdots x_{i_k}$. Denote by $\mathbf{e}_1, \ldots,  \mathbf{e}_n$ the canonical basis of $\f_2^n$, i.e. $\mathbf{e}_i$ has a 1 in position $i$ and zeroes elsewhere. Then
\begin{eqnarray}\label{eq:Moebius2}
  b_{(a_1, \ldots, a_n)} &=& \dbigoplus  _{c_1, \ldots, c_k \in \f_2} f(\dbigoplus _{j=1}^{k} c_j \mathbf{e}_{i_j})\\
  &=& \dbigoplus _{v \in V} f(v)
\end{eqnarray}
where $V = \langle \mathbf{e}_{i_1}, \ldots, \mathbf{e}_{i_k}\rangle$ is the $\f_2$-vector space generated by $\mathbf{e}_{i_1}, \ldots, \mathbf{e}_{i_k}$.

Recall that the number of subspaces of dimension $k$ of $\f_2^n$ equals the Gaussian binomial coefficient
\[
\binom{n}{k}_2 = \frac{(2^n-1)(2^{n-1}-1)\cdots (2^{n-k+1}-1)}{(2^k-1)(2^{k-1}-1)\cdots (2-1)}.
\]
Consider the general linear group $GL(n,\f_2)$, consisting of the invertible $n\times n$ matrices over $\f_2$.
 For any matrix $M\in GL(n,\f_2)$ and any $v\in \f_2^n$ we will denote by $\varphi_{M}$ and $\varphi_{M,v}$ the invertible linear, respectively affine transformation of $\f_2^n$ defined as $\varphi_{M}(x) = Mx$ and $\varphi_{M,v}(x) = Mx\oplus v$ respectively. There are $(2^n-1)(2^n-2)\cdots (2^n-2^{n-1})$ invertible linear transformations and $2^n(2^n-1)(2^n-2)\cdots (2^n-2^{n-1})$ affine ones.

Two functions $f,g:\f_2^n \rightarrow \f_2$ are called affine equivalent, denoted $f\sim g$, if $g = f \circ \varphi_{M,v}$ for some invertible affine transformation $\varphi_{M,v}$. Recall that the degree is an affine invariant, i.e. if $f\sim g$ then $\deg(f) = \deg(g)$.

Later in the paper there will be situations where only the monomials of degree $k$ or more of a polynomial are relevant, and any monomials of lower degree can be ignored. Combining that with affine equivalence, we also define the equivalence relation $f\sim_{k-1} g$ by saying that  $f\sim_{k-1} g$ if there is a function $h$ such that $f\sim h$ and $\deg(g-h)\le k-1$ (i.e. $g$ and $h$ coincide if we ignore any monomials of degree less than $k$).

\section{Degree testing  and the degree density}
We will define two notions: the ``degree less than $k$'' probabilistic test and the ``average degree-$k$ monomial density'' of a function. We will then examine the relations between them.

\begin{defn}\label{def:dt}
  Let $1\le k\le n$ be integers and let $f:\f_2^n \rightarrow \f_2$ be a function. Given $u_0, u_1, \ldots, u_k\in \f_2^n$, we will call the test \[\dbigoplus _{c_1,\ldots,c_k \in \f_2} f\left(\left(\dbigoplus _{i=1}^{k} c_iu_i\right) \oplus u_0 \right)=0\]
   the {\em degree less than $k$ test}, or $\deg(f)<k$ test.
  The probability of $f$ failing this test, taken over all $u_0, u_1, \ldots, u_k\in \f_2^n$ will be denoted $\dt_k(f)$.
In other words
\begin{equation}\label{eq:deg-test}
\dt_k(f) = \frac{|\{ (u_0, u_1, u_2, \ldots, u_k) \in (\f_2^n)^k : \dbigoplus  _{c_1,\ldots,c_k \in \f_2} f\left(\left(\dbigoplus _{i=1}^{k} c_iu_i\right) \oplus u_0 \right)  \neq 0\}|}{2^{(k+1)n}}.
\end{equation}
\end{defn}
\begin{rem}\label{rem:lin-indep}
  It is not hard to verify that if $u_1, u_2, \ldots, u_k$ are linearly dependent then any function $f$ passes that particular $\deg(f) < k$ test.
Therefore, in practice there is no need to run the test when  they are linearly dependent. We could therefore define the test either with, or without the requirement that $u_1, u_2, \ldots, u_k$ are linearly independent. We decided that both probabilities of failure, with or without the requirement
are of interest. There are at least two reasons why the probability of failure without the requirement that $u_1, u_2, \ldots, u_k$ are linearly independent, which we denoted $\dt_k(f)$, is of interest. Firstly, the case $k=2$ and $u_0=0$ corresponds to what is usually called the BLR test; the probability of failing the BLR test is defined without the requirement that $u_1,u_2$ should be linearly independent (see for example~\cite{BCHKS96}). Secondly, as we shall see in Proposition~\ref{prop:basic-prop1}(i), the value of $\dt_k(f)$ does not change if the function $f$ in $n$ variables is viewed as a function in more than $n$ variables. By contrast, this value would change if the definition required linearly independent vectors.
The probability of failing the test when we require that $u_1, u_2, \ldots, u_k$ are linearly independent, equals the quantity $\add_k(f)$ in Definition~\ref{def:add}, see Theorem~\ref{thm:relation-between-the-two-notions}.
\end{rem}

\begin{defn}\label{def:add}
  Let $0\le k\le n$ be integers and let $f:\f_2^n \rightarrow \f_2$ be a function. The degree-$k$ monomial density of $f$, denoted $\dd_k(f)$ is defined as the number of monomials of degree $k$ in the ANF of $f$, divided by $\binom{n}{k}$ (the total number of monomials of degree $k$ in $n$ variables) i.e. if the ANF of $f$ is $f(x) = \dbigoplus  _{t} b_t t$, with $t$ ranging over all monomials in $n$ variables and $b_t \in \f_2$ then
  \begin{equation*}
  \dd_k(f) = \frac{|\{t: t \mbox{ monomial of degree } k \mbox{ and } b_t\ne 0\}|}{\binom{n}{k}}.
  \end{equation*}
  The average degree-$k$  monomial density of $f$, denoted by $\add_k(f)$ is the average (arithmetic mean) of $\dd_k(g)$ of all the functions $g$ such that $f \sim g$, i.e.
  \begin{equation}\label{eq:add-def}
  \add_k(f) = \frac{\dsum_{g \sim f}\dd_k(g)}{|\{g: g \sim f \}|} = \frac{\dsum _{M \in GL(n, \f_2), v\in \f_2^n} \dd_k(f\circ \varphi_{M,v})}{2^n(2^n-1)(2^n-2)\cdots (2^n-2^{n-1})}.
  \end{equation}
\end{defn}
\begin{rem}
  The two ways of defining $\add_k(f)$ in equation (\ref{eq:add-def}) are indeed equal. Namely, denote by $A$ the cardinality of the stabilizer of $f$ under the action of invertible affine transformations of $\f_2^n$, i.e. $A = |\{\varphi_{M,v}: f\circ \varphi_{M,v} = f\}|$. Each element $g$ such that $g\sim f$ is obtained as $f\circ \varphi_{M,v}$ for $A$ distinct transformations $\varphi_{M,v}$. Therefore
  $2^n(2^n-1)(2^n-2)\cdots (2^n-2^{n-1}) = A |\{g: g \sim f \}|$ and
  $ \dsum _{M \in GL(n, \f_2), v\in \f_2^n} \dd_k(f\circ \varphi_{M,v}) = A \dsum_{g \sim f}\dd(g)$.
\end{rem}
\begin{rem}\label{rem:thickness}
  The notion of average degree-$k$ monomial density has some similarity, but is different from the algebraic thickness of a function $f$, defined in~\cite{Car02}. The algebraic thickness is defined as the minimum number of monomials among all the functions $g$ such that $f \sim g$. Both notions look at the number of monomials, but the average degree density looks at monomials of a given degree while the algebraic thickness looks at monomials of all degrees. Also, while both notions look at the whole equivalence class of $f$, the average degree density computes the average while the algebraic thickness computes the minimum.
\end{rem}

The average degree-$k$ monomial density is closely connected to the probability of failing the $\deg(f)<k$ test. After a preliminary lemma, we will give the exact relationship in the next theorem.
\begin{lem}\label{lem:lin-transf-monomial}
Let $M\in GL(n,\f_2)$ be an  invertible matrix and let $1\le i_1<\ldots < i_k \le n$. Denote by $u_1, \ldots, u_k\in \f_2^n$ the linearly independent vectors which appear in $M$ as columns $i_1 , \ldots, i_k$, respectively. Also, let $u_0\in \f_2^n$. Then the coefficient of $x_{i_1}\cdots x_{i_k}$ in $f\circ \varphi_{M,u_0}$ equals $\dbigoplus  _{c_1,\ldots,c_k \in \f_2} f\left(\left(\dbigoplus _{i=1}^{k} c_iu_i\right) \oplus u_0\right)$.
\end{lem}
\begin{proof}
Using \eqref{eq:Moebius2} and the fact that $M\mathbf{e_{i_j}} = u_j$, we see that the coefficient $b$ of $x_{i_1}\cdots x_{i_k}$ in $f\circ \varphi_{M,u_0}$ equals
\begin{eqnarray*}
  b &=&  \dbigoplus _{c_1,\ldots,c_k \in \f_2} f\left(\varphi_{M,u_0}\left(\dbigoplus _{j=1}^{k} c_j \mathbf{e_{i_j}}\right)\right) \\
   &=& \dbigoplus _{c_1,\ldots,c_k \in \f_2} f\left(M\left(\dbigoplus _{j=1}^{k} c_j \mathbf{e_{i_j}}\right)\oplus u_0\right)  \\
   &=& \dbigoplus _{c_1,\ldots,c_k \in \f_2} f\left(\left(\dbigoplus _{j=1}^{k} c_j u_j\right)\oplus u_0\right).
\end{eqnarray*}
\end{proof}
\begin{thm}\label{thm:relation-between-the-two-notions}
The average degree-$k$ monomial density of a function $f$ equals the probability of failing the test $\dbigoplus_{c_1,\ldots,c_k \in \f_2} f\left(\left(\dbigoplus_{i=1}^{k} c_iu_i\right)\oplus u_0\right) =0$ over all those $(u_0,u_1, u_2, \ldots, u_k) \in (\f_2^n)^{k+1}$ for which the vectors $(u_1, u_2, \ldots, u_k)$ are linearly independent. In other words:
   \begin{eqnarray}
 \dt_k(f)  &=&  \add_k(f) \dprod_{i=n-k+1}^n\left( 1 - \frac{1}{2^{i}}\right).
  \end{eqnarray}
\end{thm}
\begin{proof}
 We aim to count all the monomials of degree $k$ in $f \circ \varphi_{M,v}$ for all $M\in GL(n, \f_2)$ and $v \in \f_2^n$,  let us call this number $N_1$. In other words, using (\ref{eq:add-def}), $N_1$ is such that
    \begin{equation}\label{eq:add-dt-proof1}
      \add_k(f) = \frac{N_1}{\binom{n}{k}2^n(2^n-1)(2^n-2)\cdots (2^n-2^{n-1})}.
    \end{equation}

   Then we compare this number with the number of  tuples of $k+1$ vectors $(u_0, u_1, u_2, \ldots, u_k) \in (\f_2^n)^{k+1}$ for which the test fails, let us call this number $N_2$. Using (\ref{eq:deg-test}) we have
     \begin{equation}\label{eq:add-dt-proof2}
    \dt_k(f) = \frac{N_2}{2^{(k+1)n}}
    \end{equation}
    Note that the fact that the test for $(u_0, u_1, u_2, \ldots, u_k)$ fails implies that $u_1, u_2, \ldots, u_k$ are linearly independent (see Remark~\ref{rem:lin-indep}).

   For each fixed $(u_0, u_1, u_2, \ldots, u_k) \in (\f_2^n)^{k+1}$, we know from Lemma~\ref{lem:lin-transf-monomial} that the test $\dbigoplus  _{c_1,\ldots,c_k \in \f_2} f\left(\left(\dbigoplus _{j=1}^{k} c_ju_{j}\right)\oplus u_0 \right)=0$ fails if and only if the monomial $x_{i_1}\cdots x_{i_k}$ appears in $f \circ \varphi_{M,u_0}$ for some invertible matrix $M$ which has the columns $ u_1, u_2, \ldots, u_k$ in some positions $1\le i_1 < \ldots < i_k\le n$.
   The number of invertible matrices $M$ that have the vectors $u_1, u_2, \ldots, u_k$ appearing (in this order) as some of their columns is
    \[\binom{n}{k} (2^n-2^k)(2^n-2^{k+1})\cdots (2^n - 2^{n-1})\]
    (there are $\binom{n}{k}$ ways of choosing the positions where these $k$ columns appear, and for each of them we can choose incrementally the remaining $n-k$ columns such that each newly chosen column is not in the vector space generated by the previously chosen columns, to ensure that the final matrix is invertible). Therefore each fixed $(u_0, u_1, u_2, \ldots, u_k)$ for which the test fails corresponds to $\binom{n}{k} (2^n-2^k)(2^n-2^{k+1})\cdots (2^n - 2^{n-1})$ monomials in polynomials of the form $f \circ \varphi_{M,u_0}$, and thus
     \[
  N_1 = N_2 \binom{n}{d}(2^n-2^k) (2^n-2^{k+1})\cdots (2^n - 2^{n-1}).
  \]
 Combining this relation with (\ref{eq:add-dt-proof1}) and (\ref{eq:add-dt-proof2}) yields the desired result.
Note that the probability of an arbitrary $k$-tuple of vectors  $u_1, u_2, \ldots, u_k$ to be linearly independent is
\[
\frac{(2^n-1)(2^n-2)\cdots (2^n-2^{k-1})}{2^{kn}} = \left(1-\frac{1}{2^n}\right)\cdots \left(1-\frac{1}{2^{n-k+1}}\right).
\]
\end{proof}

From Lemma~\ref{lem:lin-transf-monomial} and Theorem~\ref{thm:relation-between-the-two-notions} above and the fact that the degree is invariant to invertible affine transformations, we can deduce:
\begin{cor}\label{cor:inv}
 Let $f: \f_2^n \rightarrow \f_2$ be a function.\\
  (i) If $\deg(f)<k$ then $\dt_k(f) = \add_k(f) = 0$\\
  (ii) Both $\add_k(f)$ and $\dt_k(f)$ are invariant to affine equivalence, i.e. $f \sim g$ implies $\add_k(f)= \add_k(g)$ and $\dt_k(f)= \dt_k(g)$. Moreover, $f \sim_{k-1} g$ implies $\add_k(f)= \add_k(g)$ and $\dt_k(f)= \dt_k(g)$.
\end{cor}

We prove  some useful properties of $\dt_k(f)$:
\begin{prop}\label{prop:basic-prop1}
Let $f, g_1:\f_2^n \rightarrow \f_2$ and $g_2:\f_2^m \rightarrow \f_2$  be  functions in $n$  variables.\\
(i) If $g(x_1, \ldots, x_n, x_{n+1}) = f(x_1, \ldots, x_n)$, then $\dt_k(g) = \dt_k(f)$.\\
(ii) If $g(x_1, \ldots, x_{n+m}) = g_1(x_1, \ldots, x_n) \oplus  g_2(x_{n+1}, \ldots, x_{n+m})$, then
$\dt_k(g) = \dt_k(g_1) + \dt_k(g_2) - 2 \dt_k(g_1)\dt_k(g_2)$.
\end{prop}
\begin{proof}

 (i) Consider $v_0, v_1,\ldots, v_k \in (\f_2^{n+1})^{k+1}$. Whether the $\deg(g)<k$ test passes or fails at $v_0, v_1,\ldots, v_k$ only depends on the projection of the vectors on the first $n$ coordinates, as $f$ only depends on those coordinates. Let us denote by $u_0, u_1,\ldots, u_k$, respectively,  the projections of $v_0, v_1,\ldots, v_k$ on the first $n$ coordinates. For each $u_0, u_1,\ldots, u_k$, there are $2^{k+1}$ possible values for $v_0, v_1,\ldots, v_k$ which have the same projections $u_0, u_1,\ldots, u_k$. So the number of $v_0, v_1,\ldots, v_k$ vectors on which the $\deg(g)<k$ test fails, let us call it $N_1$ equals $2^{k+1}N_2$, where we denoted by $N_2$ the number of $u_0, u_1,\ldots, u_k$ vectors on which the $\deg(f)<k$ test fails. Therefore:
 \begin{eqnarray*}
   \dt_k(g) &=& \frac{N_1}{2^{(k+1)(n+1)}} \\
    &=& \frac{2^{k+1}N_2}{2^{(k+1)(n+1)}}  \\
    &=& \frac{N_2}{2^{(k+1)n}} \\
    &=& \dt_k(f).
 \end{eqnarray*}

(ii) Since the sets of variables of the two functions $g_1$ and $g_2$ are disjoint, their values can be viewed as independent events. The function $g$ fails the test if and only if exactly one of the functions $g_1$ or $g_2$ fails the test, so $\dt_k(g) = \dt_k(g_1)(1-  \dt_k(g_2)) +  \dt_k(g_2)(1-\dt_k(g_1))$.

\end{proof}

\section{Probability of failing the $\deg(f)<k$ test when $f$ has degree $k$}

The $\deg(f)<k$ test would be used as follows: we run the $\deg(f)<k$ test a number of times, say $t$ times, for different choices of vectors $u_0,\ldots, u_k\in (\f_2^n)^{k+1}$ (chosen independently, with uniform distribution). If $f$ passes all the $t$ tests, we conclude that we probably have $\deg(f)<k$; if $f$ fails at least one of the test we conclude that $\deg(f)\ge k$.

When $\deg(f)$ is truly below $k$, the function $f$ will always pass the $\deg(f)<k$ test as, by Corollary~\ref{cor:inv}, we have $\dt_k(f)=0$. In other words, there are no false negatives. However, when the true  degree of $f$ is at least $k$, it is possible to wrongly conclude that $\deg(f)<k$ (false positive). The probability of that happening is $(1-\dt_k(f))^t$. It is therefore important to determine $\dt_k(f)$ for polynomials of degree $k$ or more.
 In this paper we commence the study of $\dt_k(f)$ by proving results for the case when $\deg(f)= k$; the case $\deg(f)> k$ will be the subject of further work.

Throughout this section we assume that $\deg(f)=k$.
Note that for affine invertible transformations $\varphi_{M,u_0}(x) = Mx+u_0$,  the monomials of maximum degree in $f \circ \varphi_{M,u_0}$, are the same regardless of the value of $u_0$, and therefore the same as the ones in $f\circ \varphi_{M}$, where $\varphi_{M}(x) = Mx$ is a linear invertible transformation. Therefore, when we study only the monomials of maximum degree, it is sufficient to look at linear, rather than affine transformations, so when $\deg(f)=k$ the equation (\ref{eq:add-def}) becomes:
 \[
  \add_k(f) = \frac{\dsum_{M \in GL(n, \f_2)} \dd(f\circ \varphi_{M})}{(2^n-1)(2^n-2)\cdots (2^n-2^{n-1})}.
  \]
Similarly, if $\deg(f) = k$ we have
\begin{equation}\label{eq:deg-test-deg-f}
\dt_k(f) = \frac{|\{ (u_1, u_2, \ldots, u_k) \in (\f_2^n)^k : \dbigoplus  _{c_1,\ldots,c_k \in \f_2} f\left(\dbigoplus _{i=1}^{k} c_iu_i\right) \neq 0\}|}{2^{kn}}.
\end{equation}

Recall from Corollary~\ref{cor:inv} that
$\add_k(f)$ and $\dt_k(f)$ are invariant under the equivalence $\sim_{k-1}$ i.e. $f\sim_{k-1} g$ implies $\add_k(f)= \add_k(g)$ and $\dt_k(f)= \dt_k(g)$. When considering polynomials of degree $k$ under $\sim_{k-1}$ equivalence, is suffices to consider representatives which only contain monomials of degree $k$, i.e. are homogeneous. In this context, the following construction was used extensively in the classification in \cite{Hou96} and \cite{LanLea07}.
Assume $f$ is homogeneous of degree $k$ and write its algebraic normal form as $f(x_1, \ldots, x_n) = \sum_{t} b_t t$, with $t$ ranging over all monomials of degree $k$ in the variables $x_1, \ldots, x_n$. Define $f^c(x_1, \ldots, x_n)  = \sum_{t} b_t \frac{x_1\cdots x_n}{t}$ (this is sometimes called the complement of $f$, but it should not be confused with the Boolean complement, which is $f\oplus 1$).

We prove  sone additional useful properties of $\dt_k(f)$:
\begin{prop}\label{prop:basic-prop2}
Let $f:\f_2^n \rightarrow \f_2$   be a  function of degree $k\ge 1$ in $n$  variables.\\
(i) If $g(x_1, \ldots, x_n, x_{n+1}) = x_{n+1} f(x_1, \ldots, x_n)$, then $\dt_{k+1}(g) = \left(1 - \frac{1}{2^{k+1}} \right) \dt_k(f)$.\\
(ii) If $f$ is homogeneous of degree $k\le \frac{n}{2}$, then $\add_{n-k}(f^c) = \add_k(f)$ and
$\dt_{n-k}(f^c) = \dt_k(f) \dprod_{i=k+1}^{n-k}\left( 1 - \frac{1}{2^{i}}\right)$.
\end{prop}
\begin{proof}

(i)   Consider $k+1$ vectors $u_0, u_1, \ldots, u_k \in (\f_2^n)^{k+1} $ such that $u_1, u_2, \ldots, u_k$ are linearly independent. Running the $\deg(f)<k$ test on these vectors involves adding the values of $f$ on all vectors in the affine space $U = u_0 \oplus \langle u_1, u_2, \ldots, u_k \rangle$. The result will be the same for all the other $(k+1)$-tuples $u'_0, u'_1, \ldots, u'_k \in (\f_2^n)^{k+1} $ for which $u'_0 \oplus \langle u'_1, u'_2, \ldots, u'_k \rangle = u_0 \oplus \langle u_1, u_2, \ldots, u_k \rangle$.
 There are $(2^k-1)(2^k-2)\cdots(2^k-2^{k-1})$ ways to choose a $k$-tuple $u'_1, \ldots, u'_k \in (\f_2^n)^{k+1} $ such that $ \langle u'_1, u'_2, \ldots, u'_k \rangle = \langle u_1, u_2, \ldots, u_k \rangle$ and $2^k(2^k-1)(2^k-2)\cdots(2^k-2^{k-1})$ ways to choose a $(k+1)$-tuple $u'_0, u'_1, \ldots, u'_k \in (\f_2^n)^{k+1} $ such that $ u'_0 \oplus\langle u'_1, u'_2, \ldots, u'_k \rangle = u_0 \oplus\langle u_1, u_2, \ldots, u_k \rangle$.

Let us denote by $L_{k+1}(g)$ the set of vector spaces of dimension $k+1$ on which the  $\deg(g)< k+1$  test fails, and
denote by $A_{k}(f)$ the set of \emph{affine}  spaces of dimension $k$ on which the  $\deg(f)< k$ test fails. We have:
\[
\dt_k(f) = \frac{2^k(2^k-1)(2^k-2)\cdots(2^k-2^{k-1})|A_k(f)|}{2^{(k+1)n}}
\]
and
\[
\dt_{k+1}(g) = \frac{(2^{k+1}-1)(2^{k+1}-2)\cdots(2^{k+1}-2^{k})|L_{k+1}(g)|}{2^{(k+1)(n+1)}}.
\]
Therefore
\[
\dt_{k+1}(g) = \dt_k(f) \left(1-\frac{1}{2^{k+1}}  \right) \frac{|L_{k+1}(g)|}{|A_k(f)|}.
\]
All we have to do now is to show that the sets $L_{k+1}(g)$ and $A_k(f)$ have the same cardinality.

Consider two projection homomorphisms: $\pi:\f_2^{n+1} \rightarrow \f_2^{n}$ defined as the projection on the first $n$ components $\pi(x_1, \ldots, x_n, x_{n+1}) = (x_1, \ldots, x_n)$ and $\psi: \f_2^{n+1} \rightarrow \f_2$ defined as the projection on the last component $\psi(x_1, \ldots, x_n, x_{n+1}) = x_{n+1}$.

Let $V\in L_{k+1}(g)$. Consider the space $V_0 =\{ v\in V: \psi(v)=0\}$. Note that $g(v)=0$ for all $v\in V_0$, as $g(x_1, \ldots, x_n, x_{n+1}) = 0$ whenever $x_{n+1}=0$. Since the test fails on $V$, $g$ must be non-zero on at least some elements of $V$. Therefore $V_0$ must be a proper subspace of $V$, and has dimension $\dim(V)-1=k$. Moreover, we can uniquely partition $V$ into the linear space $V_0$ and the affine space $V\setminus V_0 = w\oplus V_0$ where $w$ is any vector $w\in V\setminus V_0$. We have
\begin{equation}\label{eq:projection}
  \bigoplus_{v\in V} g(v) = \bigoplus_{v\in V_0} g(v)\oplus \bigoplus_{v\in w\oplus V_0} g(v) = \bigoplus_{u\in \pi(w\oplus V_0)} f(u) = \bigoplus_{u\in \pi(V\setminus V_0)} f(u).
\end{equation}
We have proved therefore that if $g$ fails on $V$ then $f$ fails on $U = \pi(V\setminus V_0)$. Conversely, for any affine space $U\in A_k(f)$, we construct the space $V\in  L_{k+1}(g)$ as follows. For any affine space $U$ of dimension $k$ there is a unique linear space $U_0$ of dimension $k$ such that $U$ can be written as $U =  u_0 \oplus U_0$ for some $u_0\in U$ (in fact any $u_0\in U$ will work). We then define $V_0$ as the $k$-dimensional linear space $\{v\in \f_2^{n+1}: \pi(v)\in U_0, \psi(v)=0\}$. We pick one vector $u_0 \in U$ and define $w$ as the vector satisfying $\pi(w)=u_0, \psi(w) =1$. Finally we define $V= V_0 \cup  (w \oplus V_0)$, which is a linear space of dimension $k+1$. Note that $V$ is the same regardless on which $u_0\in U_0$ was chosen. As in (\ref{eq:projection}) above, one can see that if $f$ fails the test on $U$ then $g$ fails the test on $V$.

(ii)
In~\cite[page 110]{Hou96}  it was proven that
the orbit of $f^c$ under the equivalence $\sim_{n-k-1}$ has the same cardinality as the orbit of $f$ under $\sim_{k-1}$, and moreover, the orbit of $f^c$ is $\{h^c: h \sim_{k-1} f\}$. Since $h$ and $h^c$ have the same number of monomials, from Definition~\ref{def:add} we have that $\add_k(f) = \add_{n-k}(f^c)$. We then apply Theorem~\ref{thm:relation-between-the-two-notions} to obtain the required result.

\end{proof}

Propositions~\ref{prop:basic-prop1} and~\ref{prop:basic-prop2} above allow us to compute the values of $\dt_k(f)$ for some simple functions $f$:
\begin{cor}\label{cor:dt-simple-functions} We have

$\dt_k(x_1\cdots x_k) = \dprod_{i=1}^k \left(1-\frac{1}{2^i}  \right)$.

$\dt_k(x_1\cdots x_k \oplus  x_{k+1}\cdots x_{2k}) = 2 \dprod_{i=1}^k \left( 1-\frac{1}{2^i}  \right) \left( 1 - \dprod_{i=1}^k \left( 1-\frac{1}{2^i}  \right)\right)$.


$\dt_{k+t}(x_{2k+1}\cdots x_{2k+t}(x_1\cdots x_k \oplus  x_{k+1}\cdots x_{2k} ) = 2 \dprod_{i=1}^{k+t} \left(1-\frac{1}{2^i}  \right) \left( 1 - \dprod_{i=1}^k \left(1-\frac{1}{2^i}  \right)\right) $.
\end{cor}

\begin{ex}
Using  the Corollary~\ref{cor:dt-simple-functions} above we compute the exact values of $\dt_k(f)$ for some functions $f$. For example $\dt_3(x_1x_2x_3) = 21/64 = 0.328125$, $\dt_2(x_1x_2 + x_3x_4) = 15/32 = 0.46875$, $\dt_4((x_1x_2 + x_3x_4)x_5x_6) = 0.384521484$. Finally, revisiting Example~\ref{ex:old-test}, for $f =x_1x_2x_3 +   x_4x_5x_6 +x_7x_8x_9$ we compute $\dt_3(f) = \frac{31437}{2^{16}} \approx 0.47969$. This means that after running 9 times the $\deg(f)<3$ test on this $f$ we have only a $(1-0.47969)^9 = 0.0028$
 probability of incorrectly deciding that $\deg(f)<3$; compare that with a probability of $0.72$ for the original test, as explained in Example~\ref{ex:old-test}.
\end{ex}

We will now prove lower and upper bounds for $\dt(f)$:

\begin{thm}\label{thm:bounds}
Let $f:\f_2^n \rightarrow \f_2$ be a function of degree $k\ge 1$. Then
\[
\dprod_{i=1}^k \left(1-\frac{1}{2^i}  \right) \le \dt_k(f) \le \frac{1}{2}\left(1-\frac{1}{2^n} \right)^{k-1}.
\]
The lower bound is achieved if and only if $f \sim_{k-1} x_1\cdots x_k $.
\end{thm}
\begin{proof} The proof will be by induction on $k$.
 We consider first the case when $k=1$. Recall that the normalised Hamming weight of  $f$ is defined as the proportion of its inputs that produce non-zero outputs, i.e.:
\[\wt(f) = \frac{|\{x\in \f_2^n: f(x)\neq 0\}|}{2^n}.\]
The $\deg(f)<1$ test is $f(u)-f(0)=0$.
  The probability of failing this test, over all $u\in \f_2^n$ is equal to $\wt(f)$ if $f(0)=0$ and it is equal to $1-\wt(f)$ if $f(0)=1$. Since $f$ has degree 1, i.e. it is an affine non-constant function, its weight is $\frac{1}{2}$. So $\dt_1(f)=\frac{1}{2}$.

 Now consider an arbitrary degree $k$ and assume the statement holds for degrees less than $k$. Recall that the discrete derivative of $f$ in a direction $a\in \f_2^n$ is defined as $\D_{a}f(x) = f(x\oplus a) \oplus f(x)$ (usually the case $a = \mathbf{0}$ is excluded, but here we will allow it, and obviously the derivative is identically zero when $a = \mathbf{0}$).

 The $\deg(f)<k$ test on $f$ at  $u_1, \ldots, u_k$
 \[
 \dbigoplus  _{c_1,\ldots,c_k \in \f_2} f(\dbigoplus _{i=1}^{k} c_iu_i) = 0
 \]
 can be rewritten as
 \[
 \dbigoplus  _{c_1,\ldots,c_{k-1} \in \f_2} \D_{u_k} f(\dbigoplus _{i=1}^{k-1} c_iu_i) = 0,
 \]
 which is the $\deg(\D_{u_k} f)< k-1$ test at  $u_1, \ldots, u_{k-1}$. We have therefore
 \begin{equation}\label{eq:induction-for bounds1}
 \dt_k(f) = \frac{1}{2^n}\sum_{u_k\in \f_2^n} \dt_{k-1}( \D_{u_k} f ).
 \end{equation}

Recall that $\deg(\D_{a})\le \deg(f)-1$ (see~\cite{Lai94}). A vector $a\in \f_2^n\setminus\{\mathbf{0}\}$ is called a \emph{fast point} for $f$ if $\deg(\D_{a})< \deg(f)-1$ (\cite{DuaLai11}). In~\cite{SalMan17} it was shown that the number of fast points for a function $f$ of degree $k$ in $n$ variables can vary from zero to at most $2^{n-k}-1$, the latter being achieved if and only if $f \sim_{k-1} x_1\cdots x_k$.
 Let us denote by $S(f)$ the vectors in $\f_2^n\setminus\{\mathbf{0}\}$ that are not fast points for $f$. We have therefore $2^n-2^{n-k} \le |S(f)| \le 2^n-1$, i.e.
 \begin{equation}\label{eq:induction-for-bounds3}
    \left(1-\frac{1}{2^k}  \right) \le  \frac{1}{2^n}|S(f)| \le  \left(1-\frac{1}{2^n}  \right).
 \end{equation}

Back to the equation~\eqref{eq:induction-for bounds1}, using the fact that whenever $u_k$ is a fast point for $f$ we have $\deg(\D_{u_k} f)<k-1$ and therefore $\dt_{k-1}(\D_{u_k} f)=0$,
we have that
\begin{equation}\label{eq:induction-for bounds2}
 \dt_k(f) = \frac{1}{2^n}\sum_{u_k\in S(f)} \dt_{k-1}( \D_{u_k} f ).
 \end{equation}
 Since $\deg(\D_{u_k} f)=k-1$ for any  $u_k\in S(f)$, by the induction hypothesis
 \[
 \dprod_{i=1}^{k-1} \left(1-\frac{1}{2^i}  \right) \le \dt_{k-1}(\D_{u_k} f) \le \frac{1}{2}\left(1-\frac{1}{2^n} \right)^{k-2}.
 \]
 Combining this with equation~\eqref{eq:induction-for bounds2} we obtain
 \[
 \frac{|S(f)|}{2^n}\dprod_{i=1}^{k-1} \left(1-\frac{1}{2^i}  \right) \le  \dt_k(f) \le \frac{|S(f)|}{2^n}\frac{1}{2}\left(1-\frac{1}{2^n} \right)^{k-2}.
 \]
and finally using~\eqref{eq:induction-for-bounds3} we obtain the bounds in the theorem's statement.

 Note that the lower bound is achieved  when $f\sim_{k-1} x_1\cdots x_k$.

\end{proof}
\begin{ex}
For functions $f$ in 8 variables, Theorem~\ref{thm:bounds} above shows that when $f$ has degree 4 we have $0.307617 \le \dt_4(f)\le 0.4941635$; when $f$ has degree 3 we have $0.328125 \le \dt_3(f)\le 0.496101379$.
\end{ex}
If we are interested in lower and upper bounds for $\dt_k(f)$ which do not depend on either $k$ or the number of variables $n$, Theorem~\ref{thm:bounds} implies:
\begin{cor} \label{cor:bounds}
Let $f$ be a function of degree $k\ge 1$. Then
\[
\prod_{i=1}^{\infty}\left(1 - \frac{1}{2^i} \right)   \le \dt_k(f) < \frac{1}{2}.
\]
The lower bound is the $q$-Pochhammer symbol at $(0.5, 0.5, \infty)$ and is equal to $0.288788\ldots$.
\end{cor}

\section{Numerical results}\label{sec:numeric}
We computed the values of $\dt_k(f)$ and $\add_k(f)$ for all functions of degree $k\ge 1$ in 8 variables. Due to the invariance to $\sim_{k-1}$, it suffices to compute these values for one representative from each class.

The cases $k=1,2$ are trivial. Namely, for $k=1$ there is only one equivalence class, with representative $f(x_1, \ldots, x_8) = x_1$, and we have $\dt_1(f) = \frac{1}{2}$ (see the first part of the proof of Theorem~\ref{thm:bounds}).
For degree $k=2$ there are four equivalence classes, corresponding to $x_1x_2, x_1x_2 + x_3x_4, x_1x_2 + x_3x_4 + x_5x_6$, and $x_1x_2 + x_3x_4 + x_5x_6 + x_7x_8$. Using Propositions~\ref{prop:basic-prop1} and~\ref{prop:basic-prop2}  we can compute $\dt_2(f)$ as being 0.375000, 0.468750, 0.492188 and  0.498047 respectively.

For degree $k=3$ we used the 31 representatives of equivalence classes of polynomials in 8 variables listed in~\cite{Hou96}. For degree $k=4$ we used the 998 equivalence classes of polynomials of degree $k=4$ in 8 variables listed in~\cite{LanLea07}.

For each function $f$ of degree $k$, we computed $\add_k(f)$ by picking one basis $u_1,\ldots,u_k$  for each of the $\binom{8}{k}_2$ vector spaces of $\f_2^8$ of dimension $k$, and then running the $\deg(f)<k$ test for each such basis.
 We then computed $\dt_k(f)$ using Theorem~\ref{thm:relation-between-the-two-notions}.

\begin{table}
\footnotesize
\caption{Test results for degree 3 in 8 variables}
\label{Hou38}
\centering
 \begin{tabular}{|l|l|l|l|}
\hline
  &  {Function $f$}   &     {$\add_3(f)$} & {$\dt_3(f)$}\\
\hline

$f_2$&$x_1x_2x_3$&0.337275&0.328125\\
$f_3$&$x_1x_2x_5\oplus x_3x_4x_5$ &0.421594&0.410156\\
$f_7$&$x_1x_2x_7\oplus x_3x_4x_7\oplus x_5x_6x_7$ &0.442674&0.430664\\
$f_4$&$x_1x_2x_3\oplus x_4x_5x_6$ &0.453213&0.440918\\
$f_5$&$x_1x_2x_3\oplus x_2x_4x_5\oplus x_3x_4x_6$&0.463753&0.451172\\
$f_6$&$x_1x_2x_3\oplus x_1x_4x_5\oplus x_2x_4x_6\oplus x_3x_5x_6\oplus x_4x_5x_6$ &0.474293&0.461426\\
$f_8$&$x_1x_2x_3\oplus x_4x_5x_6\oplus x_1x_4x_7$ &0.474293&0.461426\\
$f_{13}$&$x_1x_2x_3 \oplus  x_4x_5x_6 \oplus  x_1x_7x_8$&0.482198&0.469116\\
$f_9$&$x_1x_2x_3\oplus x_2x_4x_5\oplus x_3x_4x_6\oplus x_1x_4x_7$ &0.484833&0.471680\\
$f_{10}$&$x_1x_2x_3\oplus x_4x_5x_6\oplus x_1x_4x_7\oplus x_2x_5x_7$ &0.484833&0.471680\\
$f_{16}$&$x_1x_2x_3 \oplus  x_2x_4x_5 \oplus  x_3x_4x_6 \oplus  x_3x_7x_8$&0.484833&0.471680\\
$f_{12}$&$x_1x_2x_3\oplus x_1x_4x_5\oplus x_2x_4x_6\oplus x_3x_5x_6\oplus x_4x_5x_6\oplus x_1x_6x_7\oplus x_2x_4x_7$&0.490103&0.476807\\
$f_{14}$&$x_1x_2x_3 \oplus  x_4x_5x_6 \oplus  x_1x_7x_8 \oplus  x_4x_7x_8$&0.490103&0.476807\\
$f_{29}$&$x_1x_2x_3 \oplus  x_4x_5x_6 \oplus  x_1x_4x_7 \oplus  x_3x_6x_8$&0.490103&0.476807\\
$f_{15}$&$x_1x_2x_3 \oplus  x_2x_4x_5 \oplus  x_6x_7x_8 \oplus  x_1x_4x_7$ &0.492738&0.479370\\
$f_{11}$&$x_1x_2x_3\oplus x_1x_4x_5\oplus x_2x_4x_6\oplus x_3x_5x_6\oplus x_4x_5x_6\oplus x_1x_6x_7$ &0.495373&0.481934\\
$f_{17}$&$x_1x_2x_3 \oplus  x_1x_4x_5 \oplus  x_2x_4x_6 \oplus  x_3x_5x_6 \oplus  x_4x_5x_6 \oplus  x_1x_7x_8$ &0.495373&0.481934\\
$f_{24}$&$x_1x_2x_3 \oplus  x_1x_4x_5 \oplus  x_2x_4x_6 \oplus  x_3x_5x_6 \oplus  x_4x_5x_6 \oplus  x_1x_6x_7 \oplus  x_5x_6x_8$&0.495373&0.481934\\
$f_{28}$&$x_1x_2x_7 \oplus  x_3x_4x_7 \oplus  x_5x_6x_7 \oplus  x_2x_5x_8 \oplus  x_3x_6x_8$&0.495373&0.481934\\
$f_{31}$&$x_1x_2x_3 \oplus  x_4x_5x_6 \oplus  x_1x_4x_7 \oplus  x_3x_6x_8 \oplus  x_4x_7x_8 \oplus  x_5x_6x_8$&0.495373&0.481934\\
$f_{18}$&$x_1x_2x_3 \oplus  x_1x_4x_5 \oplus  x_2x_4x_6 \oplus  x_3x_5x_6 \oplus  x_4x_5x_6 \oplus  x_1x_6x_7 \oplus  x_2x_3x_8$&0.498008&0.484497\\
$f_{19}$&$x_1x_2x_3 \oplus  x_1x_4x_5 \oplus  x_2x_4x_6 \oplus  x_3x_5x_6 \oplus  x_4x_5x_6 \oplus  x_1x_5x_8 \oplus  x_2x_3x_7 \oplus  x_6x_7x_8$ &0.498008&0.484497\\
$f_{26}$&$x_1x_2x_3 \oplus  x_4x_5x_6 \oplus  x_1x_4x_7 \oplus  x_2x_5x_7 \oplus  x_2x_6x_8 \oplus  x_2x_7x_8 \oplus  x_3x_4x_8$ &0.498008&0.484497\\
$f_{30}$&$x_1x_2x_3 \oplus  x_4x_5x_6 \oplus  x_1x_4x_7 \oplus  x_3x_6x_8 \oplus  x_5x_7x_8$&0.498008&0.484497\\
$f_{22}$&$x_1x_2x_3 \oplus  x_2x_3x_4 \oplus  x_3x_4x_5 \oplus  x_4x_5x_6 \oplus  x_5x_6x_7 \oplus  x_6x_7x_8$ &0.499326&0.485779\\
&$\oplus  x_1x_2x_8 \oplus  x_2x_3x_8 \oplus  x_3x_4x_8 \oplus  x_4x_5x_8 \oplus  x_5x_6x_8 \oplus  x_1x_7x_8$&&\\
$f_{21}$&$x_1x_4x_5 \oplus  x_2x_4x_6 \oplus  x_3x_5x_6 \oplus  x_4x_5x_6 \oplus  x_2x_7x_8 \oplus  x_3x_4x_7 \oplus  x_1x_6x_8$  &0.500643&0.487061\\
&$\oplus  x_2x_3x_7 \oplus  x_1x_4x_7$&&\\
$f_{23}$&$x_1x_2x_3 \oplus  x_1x_4x_5 \oplus  x_2x_4x_6 \oplus  x_3x_5x_6 \oplus  x_4x_5x_6 \oplus  x_1x_6x_7 \oplus  x_5x_7x_8$&0.500643&0.487061\\
$f_{25}$&$x_1x_2x_3 \oplus  x_1x_4x_5 \oplus  x_2x_4x_6 \oplus  x_3x_5x_6 \oplus  x_4x_5x_6 \oplus  x_1x_6x_7 \oplus  x_3x_4x_8$ &0.500643&0.487061\\
$f_{32}$&$x_1x_2x_3 \oplus  x_4x_5x_6 \oplus  x_1x_4x_7 \oplus  x_1x_6x_8 \oplus  x_2x_5x_8 \oplus  x_3x_4x_8$&0.500643&0.487061\\
$f_{20}$&$x_1x_2x_3 \oplus  x_1x_4x_5 \oplus  x_2x_4x_6 \oplus  x_3x_5x_6 \oplus  x_4x_5x_6 \oplus  x_2x_7x_8 \oplus  x_3x_4x_7 \oplus  x_1x_6x_8$&0.501960&0.488342\\
$f_{27}$&$x_1x_2x_3 \oplus  x_4x_5x_6 \oplus  x_1x_4x_7 \oplus  x_2x_5x_7 \oplus  x_1x_6x_8 \oplus  x_1x_7x_8 \oplus  x_2x_4x_8 \oplus  x_3x_5x_8$  &0.503278&0.489624\\

\hline
\end{tabular}
\end{table}

Table~\ref{Hou38} lists the values of $\add_3(f)$ and $\dt_3(f)$ for each of the 31 non-zero representatives of degree 3 in 8 variables; they are listed in increasing order of $\dt_3(f)$. The values of $\dt_3(f)$ range from 0.328125 to 0.489626, with all but one polynomial (namely the one that consists of a single monomial) having $\dt_3(f)$ in the interval $[0.4, 0.5]$. As expected from Theorem~\ref{thm:bounds}, there are no values over 0.5.
We note that there are only 16 different values; some classes do have the same value of $\dt_3(f)$.

For the 998 polynomials of degree 4 in 8 variables, there are 54 different values of $\dt_4(f)$, ranging from 0.307617 to 0.480051, with all polynomials having $\dt_4(f)$ in the interval $[0.4, 0.5]$ except for two polynomial classes, namely the class of a monomial, e.g.\ $x_1x_2x_3x_4$ and the class of $x_1x_2(x_3x_4 \oplus  x_5x_6)$, which have the values 0.307617 and 0.384521, as expected from Corollary~\ref{cor:dt-simple-functions}.
As expected from Theorem~\ref{thm:bounds}, there are no values over 0.5. Histograms are given in Figures \ref{fig:Langevin48_dtf_Histogram} and \ref{fig:Langevin48Chart_dtf} in the Appendix. The first histogram shows, for equally sized intervals, the number of polynomials that have $\dt_k(f)$ in that interval. We notice that the vast majority of classes have $\dt_4(f)$ between 0.462808 and 0.480051. The second histogram shows, for each possible value of $\dt_4(f)$, the number of classes that have that particular value.  This allows us to see whether $\dt_k(f)$ could be used to distinguish between equivalence classes (recall that $\dt_k(f)$ is invariant to the equivalence $\sim_{k-1}$). Given $f,g$, if $\dt_k(f)\neq \dt_k(g)$, we know that $f$ and $g$ are inequivalent. However, if $\dt_k(f)= \dt_k(g)$ we are unable to use this invariant to decide whether $f$ and $g$ are equivalent or not. Invariants where not many classes have the same value of the invariant are therefore preferable for this purpose (several invariants have been used in the classification of~\cite{Hou96} and \cite{LanLea07}). Unfortunately, as seen from Figure~\ref{fig:Langevin48Chart_dtf}, $\dt_k(f)$ is not particularly suited for distinguishing classes, as there are many classes with the same value.

For degree $k=5$ there are again 31 representatives, which are obtained as $f^c$ with $f$ running through the 31 degree 3 representatives mentioned above. Using Proposition~\ref{prop:basic-prop2}(ii), one can compute $\dt_5(f^c) = \dt_3(f)\left(1 - \frac{1}{2^4} \right)\left(1 - \frac{1}{2^5} \right) = 0.96875\dt_3(f)$. Similarly, for degree $k=6$ there are 4 representatives, which are obtained from the degree 2 representatives by $\dt_6(f^c) = \dt_2(f)\prod_{i=3}^{6}\left(1 - \frac{1}{2^i} \right) = 0.984375\dt_2(f)$. Finally, for each of the degrees 7 and 8 there is only one class, with representatives $x_1\cdots x_7$ and $x_1\cdots x_8$, respectively. Using Corollary~\ref{cor:dt-simple-functions} we have that $\dt_k(f)$ is 0.291056
 and 0.289919, respectively.

\bibliographystyle{plain}

\begin{thebibliography}{10}

\bibitem{BCHKS96}
M.~Bellare, D.~Coppersmith, J.~H{\aa}stad, M.~Kiwi, and M.~Sudan.
\newblock Linearity testing in characteristic two.
\newblock {\em IEEE Transactions on Information Theory}, 42(6):1781--1795,
  1996.

\bibitem{Car02}
C. Carlet.
\newblock On cryptographic complexity of boolean functions.
\newblock In Gary~L. Mullen, Henning Stichtenoth, and Horacio Tapia-Recillas,
  editors, {\em Finite Fields with Applications to Coding Theory, Cryptography
  and Related Areas}, pages 53--69, Berlin, Heidelberg, 2002. Springer Berlin
  Heidelberg.



\bibitem{DinSha09}
I. Dinur and A. Shamir.
\newblock Cube attacks on tweakable black box polynomials.
\newblock In {\em EUROCRYPT}, pages 278--299, 2009.

\bibitem{DuaLai11}
M. Duan and X. Lai.
\newblock Higher order differential cryptanalysis framework and its
  applications.
\newblock In {\em International Conference on Information Science and
  Technology (ICIST)}, pages 291--297, 2011.

\bibitem{Fil02}
E. Filiol.
\newblock A new statistical testing for symmetric ciphers and hash functions.
\newblock In R. Deng, F. Bao, J. Zhou, and S. Qing, editors,
  {\em Information and Communications Security}, volume 2513 of
  {\em LNCS},
  pages 342--353. Springer, 2002.

\bibitem{Hou96}
Xiang-dong Hou.
\newblock {$GL(m, 2)$ acting on $R(r, m)/R(r-1, m)$}.
\newblock {\em Discrete Mathematics}, 149(1):99 -- 122, 1996.

\bibitem{Lai94}
Xuejia Lai.
\newblock Higher order derivatives and differential cryptanalysis.
\newblock In R.~E. Blahut, D.~J. {Costello, Jr.}, U. Maurer, and
  T. Mittelholzer, editors, {\em Communications and Cryptography}, volume
  276 of {\em The Springer International Series in Engineering and Computer
  Science}, pages 227--233. Springer, 1994.

\bibitem{LanLea07}
P.~Langevin and G.~Leander.
\newblock Classification of the quartic forms of eight variables.
\newblock In {\em Boolean Functions in Cryptology and Information Security,
  Svenigorod, Russia}, 2007.

\bibitem{McWS}
F.~J. MacWilliams and N.~J.~A. Sloane.
\newblock {\em The Theory of Error-correcting Codes}.
\newblock North Holland, Amsterdam, 1977.

\bibitem{ONe07}
S. O'Neil.
\newblock Algebraic structure defectoscopy.
\newblock Cryptology ePrint Archive, Report 2007/378, 2007.
\newblock {http://eprint.iacr.org/}.


\bibitem{SalMan17}
A. S{\u{a}}l{\u{a}}gean and M. Mandache-S{\u{a}}l{\u{a}}gean.
\newblock Counting and characterising functions with “fast points” for
  differential attacks.
\newblock {\em Cryptography and Communications}, pages 217–--239, 2017.


\bibitem{Vie07}
M.~Vielhaber.
\newblock Breaking {ONE.FIVIUM} by {AIDA} an algebraic {IV} differential
  attack.
\newblock Cryptology ePrint Archive, Report 2007/413, 2007.
\newblock {http://eprint.iacr.org/}.

\bibitem{WinSal15}
R. Winter, A. S\u{a}l\u{a}gean, and R.C.W. Phan.
\newblock Comparison of cube attacks over different vector spaces.
\newblock In J. Groth, editor, {\em 15th IMA International Conference on
  Cryptography and Coding, IMACC}, volume 9496 of {\em Lecture Notes in
  Computer Science}, pages 225--238. Springer, 2015.

\end{thebibliography}

\section*{Appendix}

\begin{landscape}
\begin{figure}[h!]
\caption{\small Histogram of $\dt_4(f)$ values of polynomials of degree 4 in 8 variables}
\label{fig:Langevin48_dtf_Histogram}
\includegraphics[scale=0.9]{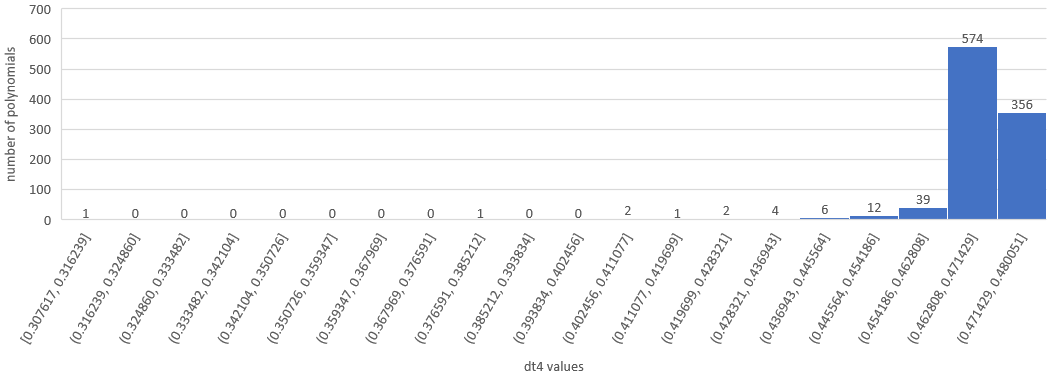}
\end{figure}
\end{landscape}

\begin{figure}[h!]
\caption{\small Number of polynomials of degree 4 in 8 variables per $\dt_4(f)$ value}
\label{fig:Langevin48Chart_dtf}
\includegraphics[scale=0.8]{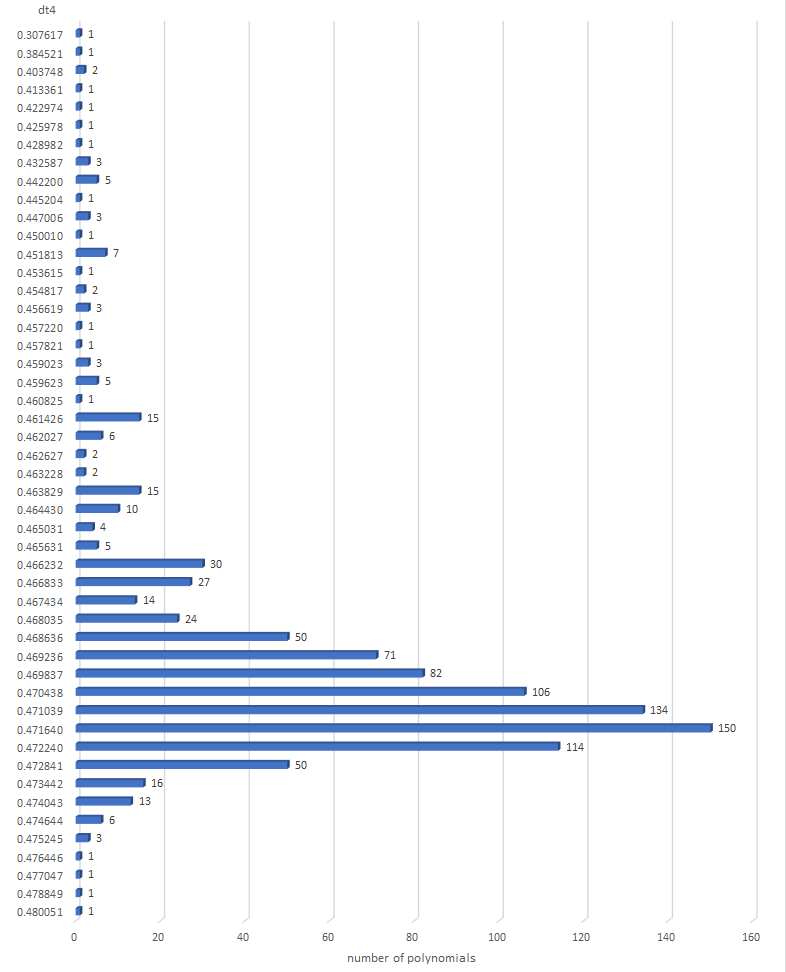}
\end{figure}

\end{document}